\documentclass[12pt]{article}

\usepackage{color}

\usepackage[margin=2.5cm]{geometry}

\usepackage{amsmath,amsthm,amsfonts,amscd,amssymb,bbm,mathrsfs, enumerate, url, ucs}
\usepackage[utf8x]{inputenc}
\usepackage{graphicx,tikz}
\usepackage[pdfborder={0 0 0}]{hyperref}
\usetikzlibrary{arrows}

\def\RR{{\mathbb R}}
\def\CC{{\mathbb C}}
\def\NN{{\mathbb N}}

\def\QQ{{\mathbb Q}}

\def\A{{\mathcal A}}
\def\B{{\mathcal B}}

\def\F{{\mathcal F}}
\def\H{{\mathcal H}}
\def\I{{\mathcal I}}
\def\K{{\mathcal K}}
\def\M{{\mathcal M}}
\def\N{{\mathcal N}}

\def\R{{\mathcal R}}
\def\S{{\mathcal S}}

\def\Z{{\mathcal Z}}

\def\a{\alpha}
\def\b{\beta}

\def\f{\varphi}

\def\i{\iota}
\def\k{\kappa}

\def\l{\lambda}

\def\r{\rho}

\def\R{{\mathcal R}}

\def\s{\sigma}

\def\t{\tau}
\def\th{\vartheta}

\def\gA{\mathfrak A}

\def\Ad{{\hbox{\rm Ad\,}}}

\def\1{{\mathbbm 1}}

\def\diff{{\rm Diff}}

\def\diffs1{\diff(S^1)}
\def\mob{{\rm M\ddot{o}b}}
\def\vir{{\rm Vir}}

\def\supp{{\rm supp\,}}

\def\psl2r{{\rm PSL}(2,\RR)}
\def\sl2r{{\rm SL}(2,\RR)}
\def\su11{{\rm SU}(1,1)}
\def\2dmob{{\overline{\psl2r}\times\overline{\psl2r}}}
\def\<{\langle}
\def\>{\rangle}

\def\Im{\mathrm{Im}\,}

\def\ln{{\mathrm{ln}}}
\def\lnu{C^*_{\mathrm{ln}}}

\DeclareMathOperator{\Tr}{Tr}

\newtheorem{theorem}{Theorem}[section]

\newtheorem{proposition}[theorem]{Proposition}
\newtheorem{lemma}[theorem]{Lemma}
\theoremstyle{remark}
\newtheorem{remark}[theorem]{Remark}
\newtheorem{example}[theorem]{Example}

\title{Rotational KMS states and type I conformal nets
}
\date{}
\author{
{\bf Roberto Longo\footnote{Supported in part by the ERC Advanced Grant 669240 QUEST ``Quantum Algebraic Structures and Models'', PRIN-MIUR, GNAMPA-INdAM and Alexander von Humboldt Foundation.} \ {\rm and} {\bf Yoh Tanimoto}\footnote{Supported by the JSPS fellowship for research abroad.}
} \\
\phantom{-}\\
Dipartimento di Matematica, Universit\`a di Roma ``Tor Vergata'' \\
Via della Ricerca Scientifica 1, 00133 Rome, Italy \\
Email: {\tt longo@mat.uniroma2.it}, {\tt hoyt@mat.uniroma2.it}
}
\begin{document}
\maketitle
\begin{abstract}
We consider KMS states on a local conformal net on $S^1$ with respect to rotations.
We prove that, if the conformal net is of type I, namely if it admits only type I DHR representations, then the extremal
KMS states are the Gibbs states in an irreducible representation.
Completely rational nets, the $\mathrm{U}(1)$-current net, the Virasoro nets
and their finite tensor products are shown to be of type I. In the completely rational case, 
we also give a direct proof that all factorial KMS states are Gibbs states.
\end{abstract}

\section{Introduction}\label{intro}
QFT, Quantum Field Theory, was originally designed to describe finitely many quantum, relativistic particles, with particle 
creation/annihilation due to the interaction. In this view, statistical mechanics aspects due to an infinitely 
many particle distribution are absent. There are however extreme situations where QFT shows a thermodynamical behaviour,
a most important one being the black hole background Hawking radiation, that lead to consider thermal states in QFT.

As is known, thermal equilibrium states at infinite volume in quantum statistical mechanics are characterized
by the KMS condition for the dynamical flow, a one-parameter automorphism group $\a_t$ of the observable $C^*$-algebra $\gA$.
A state $\f$, i.e.\! a positive linear functional on $\gA$ normalized with $\f(1) =1$, satisfies the KMS condition w.r.t.\! $\t$ 
at inverse temperature $\b>0$ if, for any $x,y\in\gA$, there is a function $F_{xy}$
analytic in the strip $S_\b=\{0<\Im z <\b\}$, bounded and continuous on the closure $\overline{S_\b}$,
such that 
\[
\begin{gathered}
F_{xy}(t)= \f\big(x\a_t(y)\big)\ , \\
F_{xy}(t+i\b)=\f\big(\a_t(y)x\big)\ ,
\end{gathered}
\]

At finite volume, where the degrees of freedom are finite,
KMS states are Gibbs states: $\f(x) = \Tr(e^{-\b H}x)/\Tr(e^{-\b H})$ with $H$ the Hamiltonian; at infinite volume,
Gibbs states might not exist as $e^{-\b H}$ is not necessarily trace class,
yet the KMS condition is preserved under the infinite volume limit.

From the mathematical viewpoint, KMS states are of most importance, being related to the Tomita-Takesaki modular theory of von Neumann algebras. The KMS condition measures, in a sense, the deviation of the state $\f$ from the tracial property $\f(xy)=\f(yx)$. 
In view of an infinite-dimensional quantum index theorem, one expects QFT to be the underlying framework and the role of the (super)-trace to be played by (super)-KMS states. The description of the KMS states then turns out a natural problem with different motivations.

This paper concerns KMS states in low dimensional CFT, Conformal Quantum Field Theory. On one hand the mathematical structure of CFT is much better understood, with very interesting connections with other mathematical subjects,
and in particular the Operator Algebraic approach 
is powerful and deep. On the other hand CFT is of much interest in Physics in various situations, e.g.\! Critical Phenomena or AdS/CFT correspondence.

CFT in $(1 + 1)$-dimensions is an extension of the tensor product of two one-dimensional
(one could say $\left(\frac12 + \frac12\right)$-dimensional) CFT, so initially one has to study
CFT on the real line or on its compactification $S^1$. The real line and the circle pictures are equivalent,
however, the physical Hamiltonian as QFT is the one associated to the translation flow in the real line picture.
The conformal Hamiltonian is the one associated with the rotation flow in the circle picture and one can usually extract
more easily information from the conformal Hamiltonian since its spectrum is discrete.

An analysis of the KMS states w.r.t.\! the translation flow has been given in \cite{CLTW12-1,CLTW12-2}.
The main result in \cite{CLTW12-1} is that, in the completely rational case, for every fixed inverse temperature $\b>0$,
there exists a unique KMS state w.r.t.\! translations, the geometric KMS state.
In the non-rational case, however, there might be uncountably many KMS states. They are all described for
the $U(1)$-current net and possibly all for the Virasoro nets \cite{CLTW12-2}.

The purpose of this paper is to investigate the KMS states with respect to the rotation flow.
In the rotational case, the first point to clarify is the choice of the $C^*$-algebra
that supports the rotational flow and on which the KMS state is to be defined.
Such a choice is natural and well known in the translation case: the $C^*$-algebra
generated by the local von Neumann algebras associated to bounded intervals of the real line.
On the other hand, the intervals of the circle do not form an inductive family
and a more thoughtful construction is necessary.
A universal $C^*$-algebras was defined by Fredenhagen, and a different construction is in \cite{Fredenhagen90, GL92}.
We shall explain in detail the construction as we need it.

We shall first give a general, complete description in the completely rational case:
every extremal KMS state is a Gibbs state in some irreducible representation.
We shall make use of the structure of the universal $C^*$-algebra in this case \cite{CCHW13};
a similar description for super-KMS states in this case is due to Hillier \cite{Hillier15}.
Rotational KMS states in the completely rational case were also studied in \cite{Iov15}.

Our results are not restricted to the rational case.
Indeed, we shall prove that any extremal rotational KMS state on a large class of non-rational conformal nets
is a Gibbs state in some irreducible representation.
The point is that, in general, the GNS representation with respect to a KMS state might
be of type II or III and could not be decomposed uniquely into irreducible (type I) representations.
We exclude this possibility for many important conformal nets.

Actually, we prove that some conformal nets are of type I, namely they do not have type II or III representations at all.
Moreover, at the moment, no example of conformal net not of type I is known.
It is possible that diffeomorphism covariance implies the type I property. 
One can understand how general the type I property is by the following.
Suppose $\A$ is a conformal net  such that, for any given $\lambda >0$,
there exists at most countably many irreducible representations $\rho$ of $\A$
such that $\lambda$ if the lowest eigenvalue of the conformal Hamiltonian
$L_0^\rho$ of $\rho$. Then $\A$ is of type I.
Many conformal nets are then immediately shown to be of type I by this criterion.
Among them are the Virasoro nets and the $U(1)$-current net. Their finite tensor products
can be shown to be of type I by a separate argument.

This paper is organized as follows.
In Section 2, we recall our operator-algebraic setting for conformal field theory and
introduce our main dynamical system, the universal $C^*$-algebra.
The fundamental examples of KMS state, the Gibbs states, are also introduced.
In Section 3, we present our classification result of KMS states.
First we are concerned with the completely rational case where the structure of the universal
$C^*$-algebra is completely understood, then we pass to the general case.
We determine that an extremal KMS state on a type I net is a Gibbs state, and
prove that some well-known nets are of type I.
The problem of the possible occurrence of type II and III representations naturally arises here and we make some observations.
In Section 4, we discuss possible applications of our results.

\section{Preliminaries}\label{preliminaries}
\subsection{Conformal nets and their representations}\label{nets}
Let us recall our mathematical framework for conformal field theory 
on the compactified one-dimensional spacetime $S^1$. See also \cite{CLTW12-1}.

Let $\I$ be the set of open, connected, non-empty and non-dense subsets (intervals) of the circle $S^1$.
A \textbf{(local) M\"obius covariant net} is a triple $(\A, U, \Omega)$ where $\A$ is a map that assigns to each $I\in\I$
a von Neumann algebra $\A(I)$ on a common Hilbert space $\H$ and satisfies the
following requirements:
\begin{enumerate}
 \item \textbf{(Isotony)} If $I_1 \subset I_2$, then $\A(I_1) \subset \A(I_2)$.
 \item \textbf{(Locality)} If $I_1 \cap I_2 = \emptyset$, then $\A(I_1)$ and $\A(I_2)$ commute.
 \item \textbf{(M\"obius covariance)} $U$ is a strongly continuous unitary representation of the M\"obius group
 $\mob = \mathrm{PSL}(2,\RR)$ on $\H$ and for any $g\in \mob$ and any interval $I\in\I$ we have
 \[
  \Ad U(g)(\A(I)) = \A(gI).
 \]
 \item \textbf{(Positivity of energy)} The generator $L_0$ of the rotation one-parameter subgroup is positive
 ($U(R_t) = e^{itL_0}$ with $R_t$ the rotation by $t$).
 \item \textbf{(Vacuum vector)} $\Omega$ is a unit vector of $\H$, which is the unique (up to a scalar)
$U$-invariant vector; $\Omega$ cyclic for $\bigcup_{I\in\I} \A(I)$.
\end{enumerate}
From these assumptions, the following automatically follow, see \cite[Section 3]{FJ96}
\begin{enumerate}
 \item[6.] \textbf{(Additivity)} If $I \subset \bigcup_{\k}I_\k$, then $\A(I) \subset \bigvee_{\k}\A(I_\k)$,
 where $\bigvee_\k \M_\k$ denotes the von Neumann algebra generated by $\{\M_\k\}$.
 \item[7.] \textbf{(Reeh-Schlieder property)} $\Omega$ is cyclic for each local algebra $\A(I)$.
\end{enumerate}

A \textbf{representation} of a M\"obius covariant net $\A$ is a family 
$\rho = \{\rho_I\}_{I\in\I}$, where $\rho_I$ is a unital $*$-representation of $\A(I)$,
on a common Hilbert space $\H_\rho$
such that $\rho_{I_2}|_{\A(I_1)} = \rho_{I_1}$ for $I_1 \subset I_2$. We say that $\rho$ is \textbf{locally normal} if
each $\rho_I$ is normal.
We say $\rho$ is \textbf{factorial} if $\bigvee_{I\in\I}\rho_I(I)$ is a factor.

A M\"obius covariant net $(\A, U, \Omega)$ is called a \textbf{conformal net}
if the representation $U$ of the M\"obius group extends to a strongly continuous projective representation
of the group $\diffs1$ of orientation-preserving diffeomorphisms of $S^1$,
that is covariant, namely $\Ad U(g)(\A(I)) = \A(gI)$, and $\Ad U(g)$
acts trivially on $\A(I)$ if $g$ is acts identically on $I$.

We say that the net $\A$ has the \textbf{split property} if for each
pair $I_1, I_2$ of intervals such that $\overline{I_1} \subset I_2$,
there is a type I factor $\N(I_1,I_2)$ such that $\A(I_1) \subset \N(I_1,I_2) \subset \A(I_2)$.
The split property follows from the conformal covariance \cite{MTW16}.

\subsection{The universal \texorpdfstring{$C^*$}{C*}-algebra}\label{universal}
For a given M\"obius covariant net $\A$, Fredenhagen \cite{Fredenhagen90} proposed to consider a $C^*$-algebra which is universal
in the sense that any representation of the net $\A$ can be regarded as a representation of this algebra.
This notion has been used widely in the study of superselection sectors in conformal field theories,
and we will take it as the algebra of our physical system.

Yet, there seems to be a confusion in the literature about the construction.
The first paper which introduced the universal $C^*$-algebra was \cite[Section 2]{Fredenhagen90}.
We take a slight variation of it: one considers the free $*$-algebra $\A_0$ generated by $\{\A(I)\}$, modulo the relations due to the inclusions $\A(I_1)\subset \A(I_2)$, for $I_1, I_2\in\I$, $I_1\subset I_2$.
Clearly a representation $\r$ of $\A$ defines a representation of $\A_0$, still denoted by $\rho$.
For a given $x \in \A_0$, one defines the seminorm by
\[
 \sup_{\rho \in \Gamma} \|\rho(x)\|,
\]
where $\Gamma$ is the class of all representations.
In the Zermelo-Fraenkel set theory with the axiom of Choice (ZFC), $\Gamma$ is not a set
(an intuitive explanation would be the following:
on each set with cardinality larger or equal to the cardinality of the continuum,
one can define a structure as a Hilbert space. Hence the class of all Hilbert spaces
is ``as large as'' the class of all sets (with cardinality larger or equal to the cardinality of the continuum),
and would cause Russell's paradox.
A precise reason is that the sets in ZFC are only those which are constructed by axiom schemas).
However, the above supremum can be justified in ZFC as follows \footnote{We owe this observation to Sebastiano Carpi.}:
For a given $x \in \A_0$, we consider the following:
\[
 \left\{s \in \RR: \text{there is a representation } \rho \text{ of } \A_0
 \text{ such that } s = \|\rho(x)\|\right\},
\]
which is a subset of $\RR$ in ZFC by the axiom schema of separation
\footnote{The axiom schema of separation reads, for a given predicate $F(x)$ with a variable $x$ as follows:
$(\exists B)(\forall x)(x\in B \leftrightarrow x \in A\; \&\; F(x))$. In words, it states that
for a set $A$ there is a subset $B$ which consists of all elements of $A$ which satisfy $F$.}
(see standard textbooks on axiomatic set theory, e.g.\! \cite{Suppes60, Jech78}).
Hence one can take the supremum and the rest follows.

Another commonly cited paper \cite[Section 8]{GL92} has a problem, because one has to take
the direct sum parametrized by ``all the representations'', which is definitely not a set.

Let us also provide a construction of the universal algebra which is closer to that of \cite{GL92}.
We consider as before the free $*$-algebra $\A_0$ generated by $\{\A(I)\}$ modulo the inclusion relations as above.
We denote by $\i_I$ the embedding of $\A(I)$ into $\A_0$.
Let $\S_0$ be the set of states (positive, unital linear functionals in the sense $\f(x^* x )\geq 0$, $x\in\A_0$) $\f$ on $\A_0$.
By definition of $\A_0$, the GNS representation $\rho_\f$ of $\A_0$ with respect to $\f$ satisfies for $I \subset J$
\[
 \rho_\f \circ \i_J|_{\A(I)} = \rho_\f\circ \i_I.
\]
Note that, for any $\f$, an element $x \in \A(I)$ is represented by a bounded operator
$\r_\f(x)$. Indeed, $x^*x \le \|x\|^2\1$ in $\A(I)$, hence $\r_\f(x^*x) \le \|x\|^2\1$ because $\r_\f |_{\A(I)}$ is a representation of a von Neumann algebra and a representation of a $C^*$-algebra is order preserving.

Now let $x$ be an element of $\A_0$; then $x$ is a finite sum $\sum_k \prod_l x_{k,l}$ of finite products of elements of $\A(I_{k,l})$, $I_{k,l}\in\I$.

Let $\rho$ be a representation of the net $\A$. Then $\r$ gives rise to a representation of $\A_0$.
With $x=\sum_k \prod_l x_{k,l}$ as above, we have
\begin{align*}
\|\r(x)\| = \left\|\r\left(\sum_k \prod_l x_{k,l}\right)\right\| &= \left\|\sum_k \prod_l \r(x_{k,l})\right\| \\
 &\leq \sum_k \left\|\prod_l \r(x_{k,l})\right\|\leq \sum_k \prod_l \|\r(x_{k,l})\|
\leq \sum_k \prod_l \|x_{k,l}\| \ ,
\end{align*}
where $\|\r(x_{k,l})\|$ is the norm of $r(x_{k,l})$ in $\A(I_{k,l})$.
Thus $ \|\r(x)\| \leq C_x < \infty$, where the constant $C_x$ does not depend on $\r$.

We define a seminorm on $\A_0$ by 
\[
 \|x\| = \sup_{\f \in \S_0}\|\rho_\f(x)\| \ ,
\]
which is finite since $\|x\| \leq C_x$,
and we take the $C^*$-completion (modulo null elements), that we denote by $C^*(\A)$. 

Let us remark that this construction avoids the set-theoretical problem:
while ``the class of all representations'' is too large to be a set,
one can consider the set of all states, because it is a subset of all maps from $\A_0$ into $\CC$
with linearity, positivity and unitarity, which can be formulated again by
the axiom schema of separation.

Now, as $C^*(\A)$ is not defined through the supremum over all representations, we have to check the universal property.
\begin{proposition}\label{pr:universal}
 For each representation $\{\rho_I\}$ of the net $\A$, there is a representation $\rho$ of the algebra $C^*(\A)$ constructed
 above such that $\rho_I = \rho \circ \i_I$.
\end{proposition}
\begin{proof}
$\{\r_I\}$ gives rise to a representation $\rho$ of $\A_0$. In order to prove that $\rho$ extends to $C^*(\A)$
we have to show that $\r$ is continuous w.r.t.\! the norm of $C^*(\A)$, namely $\|\r(x)\| \leq \|x\|$, $x\in\A_0$.
This follows because every representation of a $C^*$-algebra is direct sum of cyclic representations,
thus $\|\r(x)\|$ is the supremum of $\r_\f(x)$ with $\f$ running in a family of states.
\end{proof}

Now we may properly call $C^*(\A)$ the universal $C^*$-algebra of the net $\A$.
By the very universal property, it is unique up to an isomorphism.

Actually, we are mostly interested in locally normal representations, hence it is natural to
take account of locally normal representations only. This has been done by \cite{CCHW13}:
we take our $C^*(\A)$ and consider the locally normal universal representation  $\rho_\ln$,
which is the direct sum of all GNS representations over all states $\f$ on $C^*(\A)$ such that
$\rho_\f$ is locally normal. The universal property can be again proven by decomposing
an arbitrary representation into cyclic representations. We take the quotient
$\lnu(\A) := C^*(\A)/\ker \rho_\ln$ and call it the locally normal universal $C^*$-algebra
of the net $\A$. The properties of $\lnu(\A)$ claimed in \cite{CCHW13} can be restored
without any modification, since the actual construction is not needed in the proofs but only
the universality is used.

If the net $\A$ is conformal, any locally normal representation $\rho$ is covariant with respect
to the universal cover $\widetilde{\mob}$ of the M\"obius group
and one can take the unique implementing operators from $\rho(\lnu(\A))$, and indeed
they are finite products of local elements \cite[Theorem 6]{DFK04}.
From this it follows that the action of $\widetilde{\mob}$ on $\lnu(\A)$ is inner.

\begin{proposition}\label{pr:separable}
 Let $\A$ be a M\"obius covariant net with the split property and
 $\rho$ be a locally normal representation of $\lnu(\A)$ with a cyclic vector $\Phi$.
 Then the representation space $\H_\rho$ is separable.
\end{proposition}
\begin{proof}
 As in \cite[Appendix C]{KLM01}, we consider the set $\I_\QQ$ of intervals with rational end points,
 an intermediate type I factor $\N(I_1,I_2)$ between $\A(I_1) \subset \A(I_2), I_1,I_2 \in \I_\QQ, I_1 \subset I_2$
 (we just choose one $\N(I_1,I_2)$ for each pair $I_1 \subset I_2$, not necessarily the canonical choice of \cite{DL84}),
 let $\K(I_1, I_2)$ be the algebra of compact operators in $\N(I_1, I_2)$ (under the identification $\N(I_1, I_2) \cong \B(\H)$)
 and denote by $\mathfrak{A}$ the $C^*$-algebra generated by $\{\K(I_1,I_2)\}$. Note that $\mathfrak{A}$ is a separable
 $C^*$-algebra.
 
 As $\rho$ is locally normal, for each $I$ we have $\rho(\A(I)) \subset \rho(\mathfrak{A})''$.
 Indeed, if $I_1 \subset I_2 \subset I$, $I_1, I_2 \in \I_\QQ$, then
 $\K(I_1,I_2) \subset \A(I)$ and as $I_1$ tends to $I$, any element in $\A(I)$ can be approximated
 from $\mathfrak{A}$ in the $\s$-weak topology. Then the claim follows from the local normality of $\rho$,
 and it also follows that $\rho(\lnu(\A)) \subset \rho(\mathfrak{A})''$.
 
 Now, by assumption there is a cyclic vector $\Phi$ for $\rho(\lnu(\A))$, hence it is also cyclic for $\rho(\mathfrak{A})''$.
 As $\rho(\mathfrak{A})$ is a $C^*$-algebra, $\rho(\mathfrak{A})''$ is the closure of $\rho(\mathfrak{A})$ in the strong operator topology
 and we have $\overline{\rho(\mathfrak{A})\Phi} = \overline{\rho(\mathfrak{A})''\Phi} = \H_\rho$.
 As $\rho(\mathfrak{A})$ is separable, $\H_\rho$ is also separable.
\end{proof}
\begin{remark}
 The converse of Prop.\! \ref{pr:separable} is also true. If $\A$ is a M\"obius covariant net
 and $\rho$ a representation of $\lnu(\A)$ with separable $\H_\rho$, then $\rho$ is locally normal. 
 Indeed the $\A(I)$'s are type III factors, and every representation of a $\s$-finite type III factor
 on a separable Hilbert space is normal \cite[Theorem V.5.1]{TakesakiI}, while the local algebras
 $\A(I)$ are automatically $\s$-finite by the Reeh-Schlieder property: the vacuum state is faithful
 \cite[Proposition II.3.9]{TakesakiI}.
\end{remark}

\subsection{KMS states with respect to rotations,  Gibbs states}\label{kms}

Let $\mathfrak{A}$ be a $C^*$-algebra and $\a$ a one-parameter automorphism group of $\mathfrak{A}$
(not necessarily pointwise norm-continuous).

A \textbf{KMS state} of $\mathfrak{A}$ w.r.t. $\a$ at inverse temperature $\b \in \RR_+$
is a state $\f$ on $\mathfrak{A}$ such that for any pair of elements $x,y\in\mathfrak{A}$ there is a bounded analytic
function $F_{xy}$ on $\RR + i(0,\b)$, which is continuous on $\RR +i[0,\b]$, such that
\[
F_{xy}(t) = \f(\a_t(x)y),\quad F_{xy}(t + i\b) = \f(y\a_t(x)).
\]

Given a conformal net $\A$, we are interested in states on the universal $C^*$-algebra $C^*(\A)$ 
w.r.t. the rotation one-parameter automorphism group.
Any state $\psi$ on $C^*(\A)$ gives rise to a GNS representation $\rho_\psi$ of $C^*(\A)$,
whose restriction to $\{\A(I)\}$ (i.e. $\{\rho_\psi\iota_I\}_{I\in\I}$) is a representation of the net.
We say that $\psi$ is locally normal if its restriction to each local algebra $\A(I)$
is normal. We do not know whether this implies in general that the GNS representation $\rho_\psi$
is locally normal. Yet, for KMS states, we have the following Lemmas.
The proof of the first one is essentially the same as that of \cite[Theorem 1]{TW73}, 
one should only note that the funnel structure is not necessary.
\begin{lemma}\label{lm:locallynormalstate}
 Let  $\mathfrak{A}$ be a $C^*$-algebra
 which contains a $\s$-finite properly infinite von Neumann algebra $\M$,
 and $\f$ a state on $\mathfrak A$ such that the GNS vector $\Phi$ (for the
 GNS representation $\rho_\f$ with respect to $\f$) is separating
 for $\rho_\f(\mathfrak{A})''$. Then $\f|_\M$ is normal and $\rho_\f |_\M$ is normal.
\end{lemma}
\begin{proof}
 As $\Phi$ is separating for $\rho_\f(\mathfrak{A})''$, $\rho_\f(\mathfrak{A})''$ is $\s$-finite,
 hence $\rho_\f(\M)''$ is $\s$-finite as well.
 Then the restriction $\rho_\f$ to a properly infinite algebra $\M$ is normal
 \cite[Theorem V.5.1]{TakesakiI}.
\end{proof}

\begin{lemma}\label{lm:locallynormal}
 Let $\f$ be a KMS state on $C^*(\A)$ with respect to the rotation flow $\a$. Then $\f$
 is locally normal and its GNS representation $\rho_\f$ is locally normal.
\end{lemma}
\begin{proof}
The local algebras in the vacuum representation have a separating vector $\Omega$,
hence they are $\s$-finite, and are known to be of type III$_1$ \cite[Proposition 1.2]{GL96}.
Now the claim follows from Lemma \ref{lm:locallynormalstate} and the fact that the GNS vector $\Phi$ is separating for $\rho_\f(\A)''$ for any KMS state $\f$
 (see \cite[Lemma 5.3.8 and Corollary 5.3.9]{BR2}. The pointwise norm-continuity assumption of $\a$ is
 not necessary for this result).
\end{proof}
Thanks to these Lemmas, we do not have to distinguish $C^*(\A)$ and $\lnu(\A)$ as long as
we are interested in KMS states.

\begin{remark}
In the real line case, the GNS representation of every locally normal state 
(i.e.\! normal on each local algebra) is locally normal. To see this,
let $\f$ be a locally normal state of the quasi-local $C^*$-algebra
$\mathfrak A \equiv \overline{\bigcup_{I \Subset \RR} \A(I)}^{\|\cdot\|}$
with GNS triple $(\H,\rho,\Phi)$ and fix an interval $I\in\I$. The restriction of $\rho_I$ to $\H_I\equiv\overline{\rho(\A(I))\Phi}$ is normal as it is the GNS representation of a normal state. For any larger interval $\tilde I\supset I$ we have $\rho_I =\rho_{\tilde I}\big\vert_{\A(I)}$, so $\rho_I$ is normal on $\H_{\tilde I}$ too. 
Since the $\H_{\tilde I}$'s form an inductive family whose union is dense in $\H$ by the cyclicity of $\Phi$, it follows that $\rho_I$ is normal on $\H$.
\end{remark}

Let $\rho$ be a locally normal, rotation-covariant, irreducible representation of $\A$ in which $e^{-\b L_0^\rho}$ is trace class for $\b > 0$,
where $L_0^\rho$ is the generator of the one-parameter unitary group of rotations.
This is a typical situation that holds true in most important cases.
Then one can define the following \textbf{Gibbs state} on $C^*(\A)$:
\begin{equation}\label{Gibbs}
 \f_{\rho,\b}(x) = \frac{\Tr \big(e^{-\b L_0^\rho}\rho(x)\big)}{\Tr(e^{-\b L_0^\rho})}\ .
\end{equation}
$\f_{\rho,\b}$ is a (locally normal) rotational $\b$-KMS state of $C^*(\A)$.
Thus we have a natural class of KMS states.
The relation between the KMS condition and the Gibbs states at the given temperature can be found in \cite{Haag96}.
An early consideration of rotational Gibbs states can be found in \cite{Schroer94}.

For some important class of nets, the structure of the irreducible representations is well understood. 
This is the case, in particular,
for the class of completely rational nets, which we will consider in Section \ref{rational}.
In such cases, we shall see that all extremal KMS states are Gibbs states as in \eqref{Gibbs}.
Our main question is whether this is always true.
We show this to be true under a mild condition in Section \ref{typeIrep}, 
but the question remains open in general.

\subsection{Energy expectation value}
The stress energy density in a Gibbs state can be computed through the character formula.
For a test function $f$ with support in an interval $I\in\I$, the stress energy tensor $T$ in the vacuum representation,
smoothed with $f$, is an unbounded operator $T(f)$ affiliated to $\A(I)$. If $\r$ is an irreducible representation of $\A$,
we may define $T_\r(f) = \r(T(f))$, making use that bounded functions, e.g.\! the resolvent, of $T(f)$ belong to $\A(I)$.
The expectation value of the stress-energy tensor in the Gibbs state is then
\[
 \f_{\rho,\b}(T(f)) = \frac{\Tr (e^{-\b L_0^\rho}T_\r(f))}{\Tr(e^{-\b L_0^\rho})}\ .
\]
This is indeed finite if, for example, there is $\epsilon > 0$ such that $\Tr(e^{-(\b-\epsilon) L_0^\rho})$
is finite because the polynomial energy bound holds for the Virasoro algebra \cite[Lemma 4.1]{CW05}
(which implies that $e^{-\epsilon L_0^\rho} T(f)$ is bounded). This condition is quite generic.

Furthermore, in such a case,
one can compute this value by expanding $T_\r(f) = \sum f_nL^\r_n$,
where the $f_n =\frac1{2\pi}\int^\pi_{-\pi}f(e^{it})e^{-int}dt$ are the Fourier modes of $f$.
As $L^\r_n$, $n\neq 0$ changes the energy eigenvalues while $\Tr $
can be computed by expanding along a basis of $L^\r_0$ eigenvectors,
all the contributions from $L^\r_n, n\neq 0$ drop out and we have
\[
 \f_{\rho,\b}(T(f)) = f_0\frac{\Tr (e^{-\b L_0^\rho}L^\r_0)}{\Tr(e^{-\b L_0^\rho})} = f_0\frac{-d\chi_\rho(e^{-s})/ds|_{s=\b}}{\chi_\rho(e^{-\beta})}
 = -\frac1{2\pi}\int^\pi_{-\pi}f(e^{it})dt\cdot \frac{d\chi_\rho(e^{-s})/ds}{\chi_\rho(e^{-s})}\Big|_{s=\b}\ ,
\]
where $\chi_\rho(q) = \Tr q^{L_0^\rho}$ is known as the character of the representation $\rho$. Thus
\[
 \f_{\rho,\b}(T(1)) =
   -\frac1{2\pi}\frac{d\chi_\rho(e^{-s})/ds}{\chi_\rho(e^{-s})}\Big|_{s=\b} 
 = \frac{q}{2\pi }\frac{d\chi_\rho(q)/dq}{\chi_\rho(q)}\Big|_{q=e^{-\b}}
 = \frac{q}{2\pi }\frac{d\log(\chi_\rho(q))}{dq}\Big|_{q=e^{-\b}}.
\]
The characters for some specific examples can be found in the literature,
e.g.\! \cite{KR87}.

\section{Classification of KMS states}
\subsection{Completely rational case}\label{rational}
In this section we determine all KMS states in the completely rational case.

Let $\A$ be a M\"obius covariant net on $S^1$. Following \cite{KLM01},
one defines the $\mu$-index $\mu_\A$ of $\A$ as the Jones index of the 4-interval inclusion:
\[
\mu_\A \equiv \left[\left(\A(I_1)\vee\A(I_3)\right)' : \A(I_2)\vee\A(I_4)\right],
\]
where $I_k\in\I$, $k=1...4$, are disjoint intervals in $S^1$ whose union is dense in $S^1$ and $I_k, I_{k+2}$, $k=1,2$, have no common boundary point.

$\A$ is said to be \emph{completely rational} if $\mu_\A <\infty$ and $\A$ satisfies the split property and the strong additivity property 
\cite{KLM01}.
The split property follows from the trace class property of $e^{-\b L_0}$, for all $\b >0$
in the vacuum representation \cite{BDL07}; it holds automatically for a conformal net \cite{MTW16}.
The strong additivity property automatically holds for a conformal net with the split property and finite $\mu$-index \cite{LX04}. Thus, for a local conformal net $\A$, the only condition for $\A$ to be completely rational is $\mu_\A < \infty$. 

If $\A$ is completely rational, then 
\[
\mu_\A = \sum_k d(\rho_k)^2
\]
where $\{\rho_k\}$ is a complete family of irreducible inequivalent representations of $\A$ and
$d(\rho_k)$ is the dimension of $\rho_k$. It follows that $\A$ has only finitely many irreducible representations, all of them have finite index and 
every representation is a direct sum of irreducible finite index representations \cite{KLM01}. 

As shown in \cite{CCHW13}, the locally normal universal $C^*$-algebra $\lnu(\A)$ takes a particularly simple form in the completely rational case.

\begin{theorem}{\rm \cite{CCHW13}}\label{th:CCHW}
If $\A$ is a completely rational net, then $\lnu(\A)$ is isomorphic to a finite direct sum of type I factors:
\[
\lnu(\A) = \F_0\oplus\F_1\oplus\cdots\oplus \F_n\ ,
\]
with $\F_k= \B(\H_k)$, where $\H_k$, $k= 0,1,\dots n$ corresponds to inequivalent irreducible
representations of the net $\A$. In particular, $\lnu(\A)$ is a von Neumann algebra and its center is finite dimensional.
\end{theorem}
The minimal central projections $e_k$ of $\lnu(\A)$ are thus in one-to-one correspondence with the irreducible representations $\rho_k$ of $\lnu(\A)$:
\begin{equation}\label{sect}
\rho_k(x) = xe_k\ , \quad x\in \lnu(\A)\ ,
\end{equation}
say with $\rho_0$ the vacuum representation.

Recall that, as a completely rational net, it admits only finitely many irreducible representations (up to equivalence), 
so any representation is M\"obius covariant (see \cite[Corollary 7.2]{GL92},
and the modification to the circle is straightforward).
With $U_0$ the unitary representation of $\mob$ associated with the net $\A$,
the adjoint action of $U_0$  on the net $\A$ gives, by the universal property of $\lnu(\A)$,
an automorphism group of $\lnu(\A)$ that acts trivially on the center. In view of Theorem \ref{th:CCHW}, 
\[
U= U_0\oplus\cdots \oplus U_n\ ,
\]
with $U_k$ the covariance unitary representation of $\widetilde{\mob}$ in the representation $\r_k$.

Let now $\f$ be an extremal $\b$-KMS state of $\lnu(\A)$ w.r.t.\! the rotation one-parameter group $\a_t$.
As the GNS representation of $\f$ acts on a separable Hilbert space by Proposition \ref{pr:separable}
and Lemma \ref{lm:locallynormal}, it must either be faithful on or annihilate the components $\B(\H_k)$.
Being extremal, the support of $\f$ is $e_k$ for some $k$, namely $\f(e_j) = \delta_{jk}$.
Thus $\f$ can be viewed as a normal state on $\B(\H_k)$ and we have:

\begin{theorem}\label{th:CR}
Let $\A$ be a completely rational net as above, and $\f$ an extremal, rotational $\b$-KMS state.
Then there exists an irreducible representation $\r$ of $\A$  such that
\[
 \f(x) =\frac{\Tr \big(e^{-\b L_0^\rho}\rho(x)\big)}{\Tr(e^{-\b L_0^\rho})}\ , \quad x\in \lnu(\A)\ ,
\]
In particular $e^{-\b L_0^{\r}}$ is trace class.
\end{theorem}
\begin{proof}
By the above discussion, $\r$ is equal to a $\r_k$ given in \eqref{sect}, so the proof follows by the following lemma, which is essentially known.
\end{proof}

\begin{lemma}\label{lm:type1gibbs}
Let $\R=\B(\H)$ be a type I factor, $\th$ a one-parameter automorphism group, and $\f$ a normal $\b$-KMS state of $\R$ w.r.t.\! $\th$.
Then there exists a positive, non-singular, selfadjoint operator $H$ on $\H$ (thus affiliated to $\R$) such that
\[
 \f(x) = \frac{\Tr\big(e^{-\b H} x\big)}{\Tr (e^{-\b H})}\ , \quad x\in \R\ .
\]
We have $\Tr(e^{-\b H})<\infty $ and $\th_t(x) = \Ad e^{i t H}(x)$, $x\in \R$, $t\in \mathbb R$.
\end{lemma}
\begin{proof}
Since $\R$ is a factor, $\f$ is faithful due to the KMS property:
this follows from the faithfulness of the GNS representation and 
the fact that the GNS vector is separating for a KMS state \cite[Lemma 5.3.8 and Corollary 5.3.9]{BR2}.
As $\R$ is a type I factor, there exists a positive, non-singular trace class operator $T$ with trace one \cite[Proposition 2.4.3]{BR2}
such that $\f(x) = \Tr(Tx)$. We may write $T = e^{-\b H}$ with a self-adjoint operator $H$ and, as $T$ is bounded,
the spectrum of $H$ is bounded below.
By adding a scalar, we may assume that $H$ is positive, but then the trace $\Tr(e^{-\b H})$ is no longer $1$, so
we have the formula $\f(x) = \frac{\Tr(e^{-\b H} x)}{\Tr(e^{-\b H})}$.

Then $t\mapsto \Ad e^{-i\b t H}$ is the modular group of $\f$ \cite[Example 2.5.16]{BR2}.
Therefore, we have $\Ad e^{i t H} = \th_t$ as there is a unique one-parameter automorphism group which satisfies the KMS condition
with respect to the state $\f$ \cite[Theorem VIII.1.2]{TakesakiII}.
\end{proof}

\subsection{General case}

Let $\A$ be a M\"obius covariant net.
We say that a (locally normal) representation $\rho$ of $\A$ is of \textbf{type I} if $\rho(C^*(\A))''$ is a type I von Neumann algebra.

\subsubsection{Factorial decomposition}
\begin{proposition}\label{pr:decomposition}
Let $\A$ be a M\"obius covariant net with the split property and $\f$ a $\b$-KMS state on $C^*(\A)$ with respect
to the rotation flow $\a$. Then $\f$ can be decomposed a.e.\! uniquely as follows:
 \[
  \f = \int_X^\oplus d\mu(\l)\, \f_\l, 
 \]
 where the GNS representation $\rho_{\f_\l}$ with respect to $\f_\l$ is factorial.
 If $\rho_{\f_\l}$ is type I, then
 \[
  \f_\l(x) = \frac{\Tr\big(e^{-\beta L_0^\l}\r_{\f_\l}(x)\big)}{\Tr(e^{-\beta L^\l_0})}\ ,
 \]
 where $L^\l_0$ is the conformal Hamiltonian in the representation $\rho_{\f_\l}$.
\end{proposition}
\begin{proof}
 By Lemma \ref{lm:locallynormal}, the GNS representation $\rho_\f$ is locally normal, and
 by Proposition \ref{pr:separable} $\rho_\f$  acts on a separable Hilbert space.
 By considering the central disintegration of $\rho_\f(C^*(\A))''$,
 we also obtain  the disintegration of the representation of $\rho|_{\mathfrak{A}}$, with $\mathfrak{A}$ any separable, 
 suitably chosen $C^*$-subalgebra of $C^*(\A)$,
 by a similar argument as \cite[Proposition 56]{KLM01} (see also
 %\cite[Section II.3.1 Corollary 1]{Dixmier81},
 \cite[Theorem 8.4.2]{Dixmier77},
 \cite[Theorem IV.8.21 and Section V.1]{TakesakiI}):
 \[
  \rho_\f|_{\mathfrak{A}} = \int_X d\mu(\l)\,\rho_{\f_\l}|_{\mathfrak{A}}
 \]
 and $\rho_{\f_\l}$ are locally normal for almost all $\l$.
 According to this disintegration, the GNS vector $\Phi_\f$ disintegrates
 \[
  \Phi_\f = \int_X d\mu(\l)\,\Phi_{\f_\l}
 \]
 and the state $\<\Phi, \cdot\,\Phi\>$ on $\rho_\f(C^*(\A))''$ gets the disintegration \cite[Proposition IV.8.34]{TakesakiI}:
 \[
  \f(x) = \<\Phi_\f, \rho_\f(x)\Phi_\f\> = \int_X^\oplus d\mu(\l)\,\<\Phi_{\f_\l}, \r_\f(x)_\l\Phi_{\f_\l}\>.
 \]
 Hence we can define $\f_\l(x) = \<\Phi_{\f_\l}, \r_\f(x)_\l\Phi_{\f_\l}\>$ first for
 $x \in \mathfrak{A}$ and then extend it to $C^*(\A)$ by local normality, which is the first statement.
 $\f_\l$ are again KMS states with respect to rotations for almost all $\l$,
 by considering the disintegration of the modular operator.

 If $\r_{\f_\l}$ is of type I, then it follows that the state $\f_\l$ is given by the Gibbs state
 by Lemma \ref{lm:type1gibbs}.
\end{proof}

\subsubsection{General remarks}\label{general}

If we assume conformal covariance, type III representations do not occur
since the rotations are inner.
Furthermore, for type II states on a conformal net, a Gibbs-like formula is valid by replacing $\Tr$
by the unique tracial weight $\t$, c.f.\! Lemma \ref{lm:type1gibbs}.
\begin{lemma}\label{lm:notype3}
 If $\A$ is conformal, then for any KMS states $\f$, $\rho_\f(C^*(\A))''$ contains no type III component.
\end{lemma}
\begin{proof}
 As we saw in Section \ref{universal}, $\widetilde{\mob}$, especially the rotations, is inner.
Thus, the modular automorphisms of $\rho_\f(C^*(\A))''$ with respect to $\f$ are inner, hence the
 $\rho_\f(C^*(\A))''$ cannot have a type III component (see \cite[Theorem IV.8.21, Section V.1]{TakesakiI},
 \cite[Theorem VIII.3.14]{TakesakiII}).
\end{proof}

Let $\A$ be a conformal net and $\r$ a representation of $\A$ on $\H_\r$.
As we recalled in Section \ref{universal}, by conformal covariance,
there is a canonical inner implementation $U_\r$ on $\H_\r$ with $U_\r(g)\in \r(C^*(\A))''$ of $\widetilde{\mob}$.
The generator of the associated unitary rotation one-parameter subgroup of $U_\r$ is positive \cite{Weiner06}, which we denote
by $L_0^\r$ and we call it the conformal Hamiltonian of $\r$.
Of course, in case $\r$ is irreducible, this gives the usual definition of the conformal Hamiltonian.

\begin{lemma}\label{lem:II}
Let $\A$ be a conformal net and $\f$ an extremal, rotational $\b$-KMS state.
Suppose the GNS representation $\r_\f$ of $\f$ is of type II,
namely $\r_\f(C^*(\A))''$ is a type II factor \footnote{In this case, it would be necessarily type II$_\infty$ as the local algebras are of type III.}.
Let $\t$ denote the semi-finite trace of $\r_\f(C^*(\A))''$. 
Then there is a positive self-adjoint operator $L_0^\rho$ affiliated to $\r_\f(C^*(\A))''$ as above
and we obtain
\[
 \f(x) =\frac{\t \big(e^{-\b L_0^{\r}}\rho(x)\big)}{\t(e^{-\b L_0^{\r}})}\ , \quad x\in C^*(\A)\ ,
\]
In particular $\t(e^{-\b L_0^{\r_\l}})<\infty$.
\end{lemma}
\begin{proof}
Set $\M \equiv \r_\f(C^*(\A))''$. 
By the KMS property, the GNS vector $\xi_\f$ is cyclic and separating for $\M$ and $\Ad e^{-i\b t L_0^{\r_\l}}$
is the modular group of $\M$ w.r.t.\! to the state $\bar\f\equiv\langle\xi_\f , \cdot\ \xi_\f\rangle$ on $\M$.
By the Radon-Nikodym theorem, $\bar\f=\t(h\cdot)$ with $h$ a positive operator on $\H_\r$ affiliated to $\M$,
and $\t(h)=1$. The modular group of $\bar\f$ is then equal to $\Ad h^{it}$.
Then $h$ is proportional to $e^{-\b L_0^{\r_\l}}$, thus $h = e^{-\b L_0^{\r_\l}}/\t(e^{-\b L_0^{\r_\l}})$ and the Lemma follows.
\end{proof}

\subsubsection{Nets of type I}\label{typeIrep}
We say that a M\"obius covariant net is of \textbf{type I} if it admits only
locally normal representations $\rho$ such that $\rho(C^*(\A))''$ is of type I.

Some important conformal nets turn out to be type I,
therefore, any extremal KMS state is the Gibbs state in one of the irreducible representations.
\begin{theorem}\label{th:typeI}
If a conformal net $\A$ is of type I, then any rotational $\b$-KMS state $\f$ is a convex combination (integration)
 of the Gibbs states in irreducible representations.
\end{theorem}
\begin{proof}
 Immediate from Lemma \ref{lm:type1gibbs} and Proposition \ref{pr:decomposition}
 (note that the split property follows from conformal covariance \cite{MTW16}).
\end{proof}
As we recalled in Section \ref{universal}, for a conformal net the representatives of $\widetilde{\mob}$
are inner and unique, hence any locally normal representation $\rho$ of the net (or the universal algebra $C^*(\A)$)
is $\widetilde{\mob}$-covariant. The implementation is unique if we assume that the representatives belong to
$\rho(C^*(\A))''$. With this unique inner implementation, the lowest eigenvalue $l_0$ of the generator $L_0^\rho$
of rotations is non-negative \cite[Theorem 3.8]{Weiner06}.

\begin{proposition}\label{pr:typeI}
 Let $\A$ be a conformal net and assume that there are only countably many equivalence classes of
 locally normal irreducible representations with a specified
 lowest eigenvalue of the generator of rotations. Then $\A$ is of type I.
\end{proposition}
\begin{proof}
 By local normality and its disintegration restricted to $\mathfrak{A}$ as in Proposition \ref{pr:decomposition},
 it is enough to treat factorial representations.
 Let us consider a locally normal factorial representation $\rho$ of $\A$.
 The implementation of the $2\pi$-rotation commutes with any local element, hence with
 $\rho(C^*(\A))''$, on the other hand, it belongs to $\rho(C^*(\A))''$ by our choice that it is inner.
 When $\rho(C^*(\A))''$ is a factor, the implementation is then a scalar. This applies to any
 integer-multiple of $2\pi$, hence the spectrum of $L_0^\rho$ must be included in $\NN_0 + l_0$,
 where $l_0 \ge 0$, and $\NN_0$ is the set of non-negative integers.
 
 Now, we consider any disintegration of $\rho|_\mathfrak{A}$ into irreducible representations
 where $\mathfrak{A}$ is the separable $C^*$-subalgebra of $C^*(\A)$ as in Proposition \ref{pr:decomposition}
 (this is possible, by choosing a maximally abelian algebra in $\rho(C^*(\A))'$
 because we have the split property: see \cite[Proposition 56]{KLM01} for disintegration and
 \cite{MTW16} for the implication of the split property from conformal covariance):
 \[
  \rho_\mathfrak{A} = \int_X^\oplus \rho_\l d\mu(\l),
 \]
 where $X$ is a certain index set.
 Let us assume, by contradiction, that $\rho(C^*(\A))''$ is a factor of not type I.
 Then by \cite[Proposition 57, Corollary 58]{KLM01}, for a fixed $\l$, $\rho_\l$ is locally normal,
 hence extends to $C^*(\A)$ and must be inequivalent  to $\rho_{\l'}$ for almost all $\l'$,
 and there are uncountably many such $\l'$'s. But on the other hand,
 the inner implementation of $\widetilde{\mob}$ also disintegrates and the lowest
 eigenvalue of $L_0$ remains in $\NN_0 + l_0$ for each $\l$. By assumption, there are only countably many such inequivalent
 representations, which contradicts the above uncountable family of representations.
 This concludes the proof that $\rho$ is type I.
\end{proof}

We have two basic examples with this property.
\begin{example}\label{ex:uone}
 The $\mathrm{U}(1)$-current net $\A_{\mathrm{U}(1)}$:
 In two-dimensional spacetime, the naively defined massless free field is plagued by the infrared problem.
 Yet it is possible to consider its derivative. Its chiral components are called the $\mathrm{U}(1)$-current.
 See \cite{BMT88, Longo08} for its operator-algebraic formulation.

 The algebra is generated by the Fourier modes $\{J_n\}$ of the current which satisfy the following relations
$
  [J_m, J_n] = m\delta_{m+n, 0}
$.
 This algebra has a distinguished representation with the vacuum vector $\Omega$ such that
$
  J_m \Omega = 0 \text{ for }m \ge  0$, $ J_m^* = J_{-m}
$.
 For a smooth function $f$ on $S^1$, one defines the Weyl operator $W(f)$ by
$
  W(f) = \exp\left(i \sum_m \hat f_mJ_m\right)
$,
 where $\hat f_m$ are the Fourier components of $f(z) = \sum_m \hat f_m e^{imz}$.
 
 One defines the net by $\A_{\mathrm{U}(1)}(I) = \{W(f): \supp f \subset I \}''$. It turns out that this net is
 conformally covariant. The generator of rotations is given by the Sugawara formula
\[
  L_0 = \frac12 J_0^2 +\sum_{m>0} J_{-m}J_m.
\]
 For each $q \ge 0$, there are irreducible representations of the net $A_{\mathrm{U}(1)}$ given by the state
 $\Omega_q$ such that $J_m\Omega = 0$ for $m > 0$ and $J_0\Omega_q = q\Omega$ \cite{BMT88}.

 It can be proved that they are indeed all irreducible locally normal representations \cite{CW16}.
 By their local energy bounds, $\{J_m\}$ can be also defined in any locally normal representation. 
 In each $\rho$ of these representations, $L_0^\rho$ is again given by the above Sugawara formula
 and the lowest eigenvalue is $\frac {q^2} 2$. Namely, only two values $q$ and $-q$ share the same lowest energy.
 By Proposition \ref{pr:typeI} and Theorem \ref{th:typeI}, all KMS states with respect to rotations
 are a direct integral of Gibbs states.
 
 We also note that the regular KMS states (namely, those in whose GNS representation
 the generators $\{J_m\}$ can be defined) have been classified by \cite{BMT88}.
\end{example}

\begin{example}\label{ex:vir}
 Virasoro nets $\vir_c$: the net generated by the conformal covariance itself is called the
 Virasoro net. More precisely, one considers the group $\mathrm{Diff}(S^1)$ and its projective 
 unitary representations. There is a natural action of rotations, and if this action is also implemented
 by unitary operators and the generator is positive, then we call such a projective representation
 of $\mathrm{Diff}(S^1)$ a positive-energy representation. Such positive-energy irreducible representations
 have been classified by the so-called central charge $c > 0$ and the lowest energy $h\ge 0$
 \cite{GKO86, KR87, GW85, NS15}.
 The possible values of $c$ and $h$ are: $c = 1-\frac{6}{m(m+1)}$ and $h = \frac{((m+1)r-ms)^2-1}{4m(m+1)}$,
 where $m = 2,3,4,\cdots$ and $r = 1,2,3,\cdots, m-1$ and $s = 1,2,3,\cdots, r$, or $c\ge 1$ and $h\ge 0$.
 
 For each such a positive-energy representation $\pi_c$ with $h = 0$, one can construct the corresponding
 Virasoro net by $\vir_c(I) = \{\pi_c(g): \supp g \subset I\}''$ and it constitutes a conformal net
 (see \cite{Carpi04}). If $c < 1$, $\vir_c$ is completely rational \cite{KL04-1}.
 
 Let us consider $c \ge 1$. To any irreducible (hence type I) locally normal irreducible representation of $\vir_c$
 there corresponds a positive-energy representation of $\mathrm{Diff}(S^1)$ with $c\ge 1$ and $h \ge 0$ (see
 \cite[Proposition 2.1]{Carpi04}). Conversely, for the values $c\ge 1, h \ge 0$
 there is a corresponding locally normal irreducible representation $\pi_h^c$ of $\vir_c$
 \cite{BS90}\cite[Section 2.4]{Carpi04}\cite{Weiner16}.

 Therefore, our Proposition \ref{pr:typeI} applies to any value of $c \ge 1$ and obtain
 that any KMS state on $\vir_c$ whose GNS representation is factorial is the Gibbs state
 corresponding to the value $h$, and all such $h \ge 0$ are possible (the latter can be read off from
 the character formula, e.g.\! \cite{KR87}).
\end{example}

For a M\"obius covariant net $(\A, U,\Omega)$, one can naturally consider the tensor product
$(\A\otimes\A, U\otimes U, \Omega \otimes \Omega)$.
Any finite tensor product of these nets has again the same property. Indeed we have the following.

\begin{proposition}
 A M\"obius covariant net with the split property $\A$ is of type I
 if and only if \footnote{Actually, the split property is not necessary for the ``only if'' part.}
 any factorial locally normal representation of $\A\otimes\A$ is of the form $\rho_1\otimes\rho_2$.
\end{proposition}
\begin{proof}
 Suppose that $\A$ has only locally normal type I representations. Take a locally normal factorial representation $\tilde\rho$ of $\A\otimes\A$.
 We show that the center $\Z\left(\bigvee_{I\in\I}\tilde\rho(\A(I)\otimes\CC\1)\right)$ is trivial. Indeed,
 on one hand we have $\bigvee_{I\in\I}\tilde\rho(\A(I)\otimes\CC\1) \subset \tilde\rho(C^*(\A\otimes\A))''$.
 On the other hand,
 let us take
 \[
 p \in \Z\left(\bigvee_{I\in\I}\tilde\rho(\A(I)\otimes\CC\1)\right) = \left(\bigvee_{I\in\I}\tilde\rho(\A(I)\otimes\CC\1)\right)\cap\left(\bigvee_{I\in\I}\tilde\rho(\A(I)\otimes\CC\1)\right)'. 
 \]
 By additivity of the net and local normality of $\tilde\rho$, we have $p \in \bigvee_{I\in\I, |I| < \frac\pi 2}\tilde\rho(\A(I)\otimes\CC\1)$.
 Any element in the latter algebra commutes with $\tilde\rho(\CC\1\otimes\A(I_\k))$, where $|I_\k| < \frac\pi 2$,
 because for any pair of two intervals $I_1, I_2$ shorter than $\frac\pi 2$, one can find an interval which contains both, and
 it follows that the images $\tilde\rho(\A(I_1)\otimes \CC\1)$ and $\tilde\rho(\CC\1\otimes \A(I_2))$ commute.
 Again by additivity, $p$ commutes with $\tilde\rho(\CC\1\otimes\A(I))$ for any $I$
 and, therefore, $p \in \tilde\rho(C^*(\A\otimes\A))'$.
 Namely, $p \in \Z\left(\bigvee_{I\in\I}\tilde\rho(\A(I)\otimes\CC\1)\right) \subset \Z(\tilde\rho(C^*(\A\otimes\A))'') = \CC\1$
 as $\tilde\rho$ is factorial.
 This implies that the restriction of $\tilde\rho$ to $\A\otimes\CC\1$ is already factorial, and by assumption,
 it is of type I, namely its image is of the form $\B(\H_1)\otimes\CC\1$, where $\H_{\tilde\rho} = \H_1\otimes\H_2$.
 As the image $\bigvee_{I\in\I}\tilde\rho(\CC\1\otimes\A(I))$ commutes
 with $\bigvee_{I\in\I}\tilde\rho(\A(I)\otimes\CC\1) = \B(\H_1)\otimes\CC\1$
 by the same argument as above,
 we have $\bigvee_{I\in\I}\tilde\rho(\CC\1\otimes\A(I)) \subset \CC\1\otimes\B(\H_2)$. In other words, $\tilde\rho$ is a product representation
 of the form $\rho_1\otimes\rho_2$.
 
 To show the converse under the split property, we take a non-type I factorial representation $\rho$ of $\A$
 and construct $\bar \rho(x) = J_\rho\rho(JxJ)J_\rho$, where
 $J_\rho$ is the modular conjugation of $\rho(\A)''$ with respect to a certain faithful normal weight
 and $J$ is an antilinear conjugation which maps local algebras to local algebras,
 for example, the modular conjugation of an interval with respect to the vacuum state.
 We define $\tilde \rho(x\otimes y) = \rho(x)\bar\rho(y)$.
 We first show that this is a locally normal representation.
 By the split property, $\A(I_1)$ is included in a type I factor $\N(I_1,I_2)$, where $\overline{I_1}\subset I_2$.
 As $\rho$ is locally normal,
 the image $\tilde\rho(\N(I_1,I_2)\otimes\CC\1) = \rho(\N(I_1,I_2))$ is again a type I factor.
 The image $\tilde\rho(\CC\otimes\A(I_1))$ commutes with this, therefore, $\tilde\rho$ is locally
 a tensor product representation, and therefore, locally normal.
 The consistency condition for $\tilde\rho$ is obvious, hence it is a locally normal representation of $\A\otimes\A$.
 Yet its image is $\B(\H_\rho)$, and its restriction is not of type I, thus $\tilde\rho$ cannot be a product representation.
\end{proof}

\subsection{Remarks on non-type I representations: open problems}

As we saw, important classes of conformal nets are of type I.
It is an open problem whether there exists a M\"obius covariant net not of type I.
The situation is quite different from the case of nets on the real line,
where any translation KMS state on the quasilocal algebra is of type III$_1$ \cite{CLTW12-1} or of nets on the Minkowski space 
where one can have any type of representation \cite{DS82, DS83, BD84, DFG84}.

One concrete open case is the cyclic orbifold \cite{LX04}. Take a conformal net $\A$, make the tensor product $\A\otimes \A$
and consider the fixed point net $(\A\otimes\A)^\mathrm{flip}$ with respect to the flip between two components.
If $\A$ is completely rational, then $(\A\otimes\A)^\mathrm{flip}$ is again completely rational and
all the sectors can be explicitly written in terms of sectors of $\A$ and twisted sectors.
On the other hand, if $\A$ is not completely rational, then we do not have a complete classification
of sectors of $(\A\otimes\A)^\mathrm{flip}$. In particular, we are not able to exclude the possibility of
non-type I representations, although all known sectors are of type I.

Another candidate for a net with non-type I representations would be an infinite tensor product.
Recall that (see e.g.\! \cite[Section 6]{CW05}) for a given countable family of M\"obius covariant nets $\{(\A_k, U_k, \Omega_k)\}$,
one can define a M\"obius covariant net by
\[
 \A(I) := \bigotimes \A_k(I),\; U(g) := \bigotimes U_k(g),
\]
with respect to the reference vector $\Omega = \bigotimes \Omega_k$ (see e.g.\! \cite[Section XIV.1]{TakesakiIII}).
Let us assume that each $\A_k$ admits a representation $\rho_k$ which is converging to the vacuum representation
in some sense. Then one may hope that the infinite tensor product of representations
$\bigotimes \rho_k$ could make sense. Even if each $\rho_k$ is of type I,
the resulting product could be of non type I. However, this discussion depends on the nature of the sequence $\rho_k$
and a detailed analysis is needed.

We note that the type I property for rotational $\b$-KMS states can be characterized by a compactness criterion similar
to the Haag-Swieca compactness condition (and the Buchholz-Wichmann nuclearity condition, see \cite{Haag96}). 
\begin{proposition}
Let $\A$ be a local conformal net and $\f$ a rotational, factorial $\b$-KMS state of $C^*(\A)$.
Then $\rho\equiv\rho_\f$ is of type I if and only if the closure of $e^{-\frac \b 4 L_0^{\r}}\rho_\f(C^*(\A)_1)\Psi$
is compact in the norm topology of $\H$ for some, hence for every, non-zero vector $\Psi$ of the GNS Hilbert space $\H$ of $\f$.
Here the suffix 1 denotes the unit ball.

In this case $\overline{e^{-s L_0^{\r}}\rho_\f(C^*(\A)_1)\Psi}$ is compact for every $s>0$, $\Psi\in\H$.
\end{proposition}
\begin{proof}
Let $\M$ be the weak closure of $\rho(C^*(\A))$. By assumption $\M$ is a factor.
Moreover, $\overline{\M_1\Psi} = \overline{C^*(\A)_1\Psi}$ by Kaplansky density theorem.
Note that $e^{-s L_0^{\r}}\in \M$ for $s > 0$. Let $T_\Psi^{(s)}  : \M \to \H$ be the map $x\mapsto e^{-s L_0^{\r}}x\Psi$.
Clearly $T^{(s)}_\Psi$ is compact if and only if $\overline{e^{-s L_0^{\r}}\M_1\Psi}$ is compact.
Now if $\overline{e^{-s L_0^{\r}}\M_1\Psi}$ compact, then $\overline{e^{-s L_0^{\r}}\M_1\Psi'}$ is compact
for any other vector $\Psi'$ in the linear span of $\{xx'\Psi : x\in \M, x'\in \M'\}$,
which is a dense subspace of $\H$ as $\M$ is a factor.
Since $\|T_{\Psi_1} - T_{\Psi_2}\| = \|T_{\Psi_1-\Psi_2}\| \leq \|e^{-s L_0^\r}\|\, \|\Psi_1 - \Psi_2\| \leq \|\Psi_1 - \Psi_2\|$ 
($\|e^{-s L_0^\r}\|\leq 1$ as $L_0^\r$ is positive), $T^{(s)}_{\Psi'}$ is then compact for all $\Psi'\in\H$.

Assume first that $\M$ is of type I. As $\M$ is in the standard form, we may identify $\M \simeq \B(\K)\otimes \CC\1$ where $\H = \K\otimes \overline{\K}$,
and further $\H$ with the Hilbert space $\mathrm{HS}(\K)$ of the Hilbert-Schmidt operators,
so every vector $\Psi\in\H$ with a Hilbert-Schmidt operator $S$. In this identification, an element $x\in \M$ acts
by left multiplication on $\mathrm{HS}(\K)$. Thus 
$e^{-s L_0^{\r}}\M_1\Psi$ is identified with  $e^{-s L_0^\r}B(\K)_1S$,
whose closure is compact because $e^{-\b L_0^\r}$ is of trace class (Lemma \ref{lm:type1gibbs}), hence $e^{-s L_0^\r}$ is compact
for any $s > 0$ hence for $s = \frac\b 4$,
thus $x\in\B(\K)\mapsto e^{-\frac\b 4 L_0^\r}xS$ is compact
(c.f.\! \cite{BDL90}).

On the other hand, assume now that $T^{(\frac\b 4)}_\Psi$ is compact for some non-zero $\Psi$, then by the above argument $T^{(\frac\b 4)}_{\Phi'}$
is compact for any vector $\Phi'$. As the rotation one-parameter group is inner,
the modular operator $\Delta$ of $\M$ w.r.t.\! $\Psi$ is given by
$\Delta = e^{-\b L_0^{\r}}Je^{\b L_0^{\r}}J$ with $J$ the modular conjugation of $(\M,\Phi)$.
Thus the map $x\in \M\mapsto e^{-\frac{\b}4 L_0^{\r}} Je^{\frac{\b}4 L_0^{\r}}J  x\Phi = T^{(\frac{\b}{4})}_{\Phi'}(x)$,
with $\Phi' = Je^{\frac{\b}4 L_0^{\r}}J\Phi$ is compact ($\Phi$ belongs to the domain of $Je^{s L_0^{\r}}J$ if $s<\frac \b 2$).
As this is the modular nuclearity map $x\in \M\mapsto \Delta^{\frac 14}x\Phi\in\H$, $\M$ is of type I by \cite[Corollary 2.9]{BDL90}.
\end{proof}

\section{Outlook}
Although the conformal Hamiltonian is not the physical Hamiltonian, namely it
does not implement the time translation QFT flow, 
there is some physical interest in considering rotational KMS states in CFT.

One example comes from the three-dimensional quantum gravity.
If the cosmological constant is assumed to be negative, one should then look at the solutions
of the Einstein equation which are asymptotically close to the $\mathrm{AdS_3}$ spacetime. Different solutions
have different boundary data and such solutions (with certain fall-off conditions)
have been classified in \cite{GL14}. Two copies of the Virasoro group make the transformations
between these solutions. Such an action of the Virasoro group
is called a coadjoint action \cite{Witten88}. Maloney and Witten \cite{MW10} tried to compute the partition
function of the $\mathrm{AdS_3}$ gravity, but they arrived at an expression which cannot be interpreted
as a trace over a Hilbert space of the exponential of a self-adjoint operator. It has been proposed
to study each orbit of the Virasoro group first, e.g.\! \cite{GL14}. In particular, one can consider the so-called
BTZ black hole solutions \cite{BTZ92}. In a hypothetical quantum theory, the Virasoro group should appear
as a symmetry of the theory, while the black hole should be in a thermal state. 
Furthermore, the energy, hence the mass, of the black hole corresponds to the conformal Hamiltonian
(see \cite[eq.\! (50)]{GL14}).
In this way, KMS states on the Virasoro nets with respect to rotations
should appear naturally. From our results, one can conclude that
all such KMS states can be represented on the direct sum or integral of the Verma module.

Besides, it is an interesting question to make sense of quantum entropy of such black hole states
from the operator-algebraic point of view.

\subsubsection*{Acknowledgement}
We thank Sebastiano Carpi for his remark on the axiomatic set theory,
Alan Garbarz and Mauricio Leston for inspiring discussions on three-dimensional
quantum gravity and Mih\'aly Weiner for informing us of the classification of irreducible sectors
of the $\mathrm{U}(1)$-current net.

R.L. thanks the organisers of the ``von Neumann algebras'' program at the Hausdorff Institute in Bonn for the hospitality extended to him during May and July 2016, when part of this paper has been written.

{\small
\def\cprime{$'$} \def\polhk#1{\setbox0=\hbox{#1}{\ooalign{\hidewidth
  \lower1.5ex\hbox{`}\hidewidth\crcr\unhbox0}}}
 
}

\end{document}